\documentclass[reprint,aps,superscriptaddress]{revtex4-1} %,preprint

\usepackage[pdftex]{graphicx}
\usepackage{amsmath}
\usepackage{amsthm}
\usepackage{color}

\newcommand{\ketbra}[2]{|#1\rangle\langle#2|}

\def\tr{\textrm{tr}}
\def\wt{\widetilde}
\def\cp{{\cal P}}
\def\lw{{\cal L}_u}
\def\cc{{\cal C}}
\def\cq{{\cal Q}}
\def\cn{{\cal N}}

\def\ce{{\cal E}}

\def\cR{{\cal R}}

\newtheorem{theorem}{Theorem}
\newtheorem{lemma}{Lemma}

\newtheorem{conjecture}{Conjecture}

\def\squareforqed{\hbox{\rlap{$\sqcap$}$\sqcup$}}
\def\qed{\ifmmode\squareforqed\else{\unskip\nobreak\hfil
\penalty50\hskip1em\null\nobreak\hfil\squareforqed
\parfillskip=0pt\finalhyphendemerits=0\endgraf}\fi}

\begin{document}
\title{Superadditivity of private information for any number of uses of the channel}

\newcommand{\DAMTP}{Department of Applied Mathematics and Theoretical Physics, University of Cambridge, Cambridge CB3 0WA, U.K.}

\author{David Elkouss}
\affiliation{Departamento de An\'alisis Matem\'atico and Instituto de Matem\'atica Interdisciplinar, Universidad Complutense de Madrid, 28040 Madrid, Spain}
\author{Sergii Strelchuk}
\affiliation{\DAMTP}

\begin{abstract}
The quantum capacity of a quantum channel is always smaller than the capacity of the channel for private communication. 
However, both quantities are given by the infinite regularization of respectively the coherent and the private information. 
Here, we construct a family of channels for which the private and coherent information can remain strictly superadditive for unbounded number of uses. We prove this by showing that the coherent information is strictly larger than the private information of a smaller number of uses of the channel. This implies that even though the quantum capacity is upper bounded by the private capacity, the non-regularized quantities can be interleaved. From an operational point of view, the private capacity can be used for gauging the practical value of quantum channels for secure communication and, consequently, for key distribution. We thus show that in order to evaluate the interest a channel for this task it is necessary to optimize the private information over an unlimited number of uses of the channel.

\end{abstract}
\maketitle

%\tableofcontents

%%%%%%%%%%%%%%%%%%%%%%
%\section{Introduction}
%%%%%%%%%%%%%%%%%%%%%%

How well is it possible to characterize the resources available to transmit information? In classical information theory, this proves to be fully within our computational abilities: given a description of a channel, answering the question about its capacity to convey information to the receiver is straightforward. However, our world is inherently quantum and when one turns to the channels which transmit quantum information -- the amount of resources required to compute their capacities is unknown at best. To compute a number of different types of capacity of the quantum channel, defined as regularized quantities~\cite{Lloyd_97,Shor_02,Devetak_05}, it is necessary to perform an unbounded optimization over the number of the copies of the channel. The action of a channel $\cn^{A\to B}$ can be defined via an isometry $V^{A\rightarrow BE}$: $\cn^{A\to B}(\rho)=\tr_E V \rho V^*$, and its complementary channel is $\cn^{A\to E}_c(\rho)=\tr_ B V \rho V^*$. In the following, we will omit the register superscripts when it does not add to clarity. 

The quantum and classical capacity of a channel are given by:
\begin{align}
\cq(\cn) &= \lim_{n\rightarrow\infty}\frac{1}{n} \cq^{(1)}(\cn), \\
\cc(\cn) &= \lim_{n\rightarrow\infty}\frac{1}{n} \cc^{(1)}(\cn) \label{eq:ccapacity}
\end{align}
where
\begin{align}
\cq^{(1)}(\cn) &= \max_{\rho^A} H(B) - H(E) , \\
\cc^{(1)}(\cn) &=\max_{\rho\in\cal R} I(X;B) .
\end{align}
The optimization of the quantum capacity is performed over all valid states on the input register $A$ while the optimization of the classical capacity is performed over $\cR$ the set of classical-quantum states of the form $\rho^{XA} = \sum_x p_x \ketbra{x}{x}^X \otimes \rho_x^A$. Where $X$ is an auxiliary classical register, $H$ is the von Neumann entropy and $I(X;B)$ is the quantum mutual information. 

From the above expressions it follows that one has to optimize over an {\it infinite} number of copies of the channel in order to compute its capacity.
Do we have to resort to the regularized expression in order to compute the capacity of a quantum channel? It has recently been shown that at least in the case of the quantum capacity this is unavoidable~\cite{DiVincenzo_98,Smith_07} even when we attempt to answer the question of whether the channel has any capacity at all~\cite{Cubitt_14}. For the classical capacity, which is known to be superadditive for two uses of the channel~\cite{Hastings_09}, there is some evidence that ultimately the regularization might not be required~\cite{Montanaro_13,Smith_13}. 

Arguably, the biggest practical success of quantum information theory to date is the possibility of quantum key distribution (QKD). QKD allows two distant parties to agree on a secret key independent of any eavesdropper. The required assumptions are: access to a quantum channel with positive private capacity and the validity of quantum physics \footnote{In order to characterize the channel and to implement a specific QKD protocol one might need a public authentic classical channel.}. On the other hand, key distribution is a primitive that can only be implemented with classical resources if one is willing to constrain the power of the eavesdropper. Even though there exist practical QKD schemes which enable secure communication over large distances with high key rates~\cite{Comandar_14, Korzh_14, Jouguet_13, Shimizu_14}, some of the fundamental questions about the capacity to transmit secure correlations remain unanswered. 

There are essentially two quantities which describe the ability of the channel to send secure messages to the receiver. The first one, is called private capacity $\cal P$. It has a clear operational interpretation as the maximum rate at which the sender, Alice, can send private {\it classical} communication to the receiver, Bob. It is defined as follows:
\begin{equation}\label{privatecapformula}
\cp(\cn) = \lim_{n\rightarrow \infty}\frac{1}{n}\cp^{(1)}(\cn^{\otimes n}).
\end{equation}
That is the private capacity is given by the regularization of ${\cp}^{(1)}(\cn)$, the private information of the channel, which is given by 
\begin{equation}
\cp^{(1)}(\cn) = \max_{\rho\in\cal R} I(X;B) - I(X;E).
\end{equation}
\noindent One can view private capacity as the optimal rate of reliable communication keeping Eve in a product state with Alice and Bob. 

This capacity characterizes the optimal rates of QKD. A better understanding of this quantity would allow to evaluate precisely the usefulness of communications channels for practical QKD links. For instance, the private capacity of Gaussian channels \cite{Weedbrook_12} remains open. Beyond the pure loss channel \cite{Wolf_07} only lower bounds on the private information of a single use are known.

In the case of private capacity, the eavesdropper, Eve, is given a purification of the channel output which means that she is as powerful as it is allowed by quantum mechanics. However, this setting may be too restrictive for practical applications given the current state of the art in quantum information processing. A natural relaxation of this strong security requirement is to assume that Eve obtains information about the key by performing a measurement on her state. This security requirement is reflected in the second quantity, locking capacity $\mathcal L$. By $\mathcal L$ we denote all the recently introduced locking capacities~\cite{Guha_14} of a quantum channel. They are defined by the optimal rate of reliable classical communication requiring Eve to have vanishing accessible information about the message. This difference in the security criterion has striking consequences. For instance it implies that some channels that have no private capacity have close to maximum locking capacity \cite{Winter_14} and for some relevant classes of channels locked communications can be performed at almost the classical capacity rate \cite{Lupo_14}. The following upper bound is known for the locking capacities:
\begin{equation}\label{wlockingcapformula}
\mathcal L(\cn)\leq\lw(\cn) = \sup_n \frac{1}{n}\lw^{(1)}(\cn^{\otimes n})
\end{equation}
where $\lw^{(1)}$, that we will call the locking information, is given by: 
\begin{equation}
\lw^{(1)}(\cn) = \max_{\rho\in\cal R} I(X;B) - I_\textrm{acc}(X;E).
\end{equation} 
The \textit{accessible information} $I_\textrm{acc}(X;E)=\max_{\Gamma}I(X;Y)$, where $\Gamma$ is the set of all POVMs on $E$. 

Despite the significance of the private and locking information, we still understand very little about its behaviour when the communication channel is used many times. Authors in~\cite{Smith_08b,Kern_08} provide evidence that $\cp^{(1)}(\cn)$ is superadditive for a small finite number of channel uses, although the magnitude of this effect is quantitatively very small. Recently, it has been shown the existence of two quantum channels $\cn_1, \cn_2$ with $\cc(\cn_1)\le2, \cp(\cn_2)=0$ for which $\cp(\cn_1\otimes\cn_2)\ge 1/2\log d$, where $d$ is the dimension of the output of the joint channel~\cite{Smith_09b}. This example shows that the private capacity is a superadditive quantity (this was also proved in \cite{Li_09} using a different construction). 

Even less is known about the locking capacity. It follows trivially that $\lw^{(1)}$ is sandwiched between the classical information and the private information \cite{Guha_14}:
\begin{equation}
\cp^{(1)}(\cn) \leq \lw^{(1)}(\cn) \leq \cc^{(1)}(\cn).
\end{equation}

Here we show that private information can be strictly superadditive for an arbitrarily large number of uses of the channel. More precisely, we prove the following theorem:

\begin{theorem}\label{mainthm}
For any $n$ there exists a quantum channel $\cn_{n}$ such that for $n>k\geq 1$:
\begin{equation}
\frac{1}{k}\cp^{(1)}(\cn_{n}^{\otimes k})<\frac{1}{k+1}\cq^{(1)}(\cn _{n}^{\otimes k+1}).
\end{equation}
\end{theorem}

This proves that entangled inputs increase the private information of a quantum channel and this effect persists for an {\it arbitrary} number of channel uses. As a bonus, we obtain a qualitatively different proof for the unbounded superadditivity of the coherent information~\cite{Cubitt_14}.

We now introduce the key components of our construction which are required to prove Theorem~\ref{mainthm}. 

\textbf{Main construction: switch channel.} 
We first introduce \textit{switch channels}:
\begin{equation}\label{switchchn}
\cn^{SA\to SB}(\rho^{SA})=\sum_i P^{S\rightarrow S}_i\otimes \cn^{A\rightarrow B}_i (\rho^{SA}).
\end{equation}
A switch channel consists of two input registers $S$ and $A$ of dimensions $d$ and $n$ respectively. Register $S$ is measured in the standard basis and conditioned on the measurement outcome $i$ a \textit{component} channel $\cn_i$ is applied to the second register. 
The computation of $\cp^{(1)}(\cn)$ and $\lw^{(1)}(\cn)$ when $\cn$ is of the form~\eqref{switchchn} can be simplified; it suffices to restrict inputs to a special form. The equivalent result for the quantum capacity was proved in \cite{Fukuda_07}.
\begin{lemma}
Consider a switch channel $\cn^{SA\rightarrow SB}$ and let ${\cal T} = \{\rho : \rho = \sum _x p_x \ketbra{x}{x}^X\otimes \ketbra{s}{s}^S\otimes\rho_x^A\}$. Then:
\begin{enumerate}
\item $\cp^{(1)}(\cn) =\max_{1\le s<n} \cp^{(1)}(\cn_s)$,
\item $\lw^{(1)}(\cn) = \max_{1\le s<n} \lw^{(1)}(\cn_s)$.
\end{enumerate}
Both $\cp^{(1)}(\cn)$ and $\lw^{(1)}(\cn)$ can be achieved by some $\rho\in\cal T$.
\end{lemma}
%{\it Proof.} 
\begin{proof}
The channel complementary to a switch channel is also a switch channel with component channels $\{{\cn}_i^c\}_{i=1}^n$ complementary to $\{{\cn_i}\}_{i=1}^n$ \cite{Cubitt_14}. We denote the output systems of the complementary channel by $S$ and $E$. Let $\rho\in\cal R$ be the input state that maximizes $\lw^{(1)}(\cn)$. The following chain of inequalities holds:
\begin{align}
I(X&;BS) - I_\textrm{acc}(X;ES)\\ 
            &= \sum_s p_s\Big( I(X;B|S=s) - I_\textrm{acc}(X;E|S=s)\Big) \\
            &\leq \max_s \Big( I(X;B|S=s) - I_\textrm{acc}(X;E|S=s)\Big)\\
            &\leq \max_s \lw^{(1)}(\cn_s).
\end{align}
The first equality follows because $S$ is a classical system. The first inequality follows by choosing the value of $s$ which maximizes the difference between the mutual and the accessible information. The second one since the difference between the between the mutual information to the receiver and the accessible information to the environment is upper bounded by the locking information of the channel $\cn_s$. This upper bound is achievable by an input state of the form $\sigma^{XSA}=\sum_x p_x \ketbra{x}{x} \otimes \ketbra{s}{s} \otimes \rho_x$ where $\tr_S(\sigma^{XSA})$ is the state that achieves the locking information of channel $\cn_s$. Finally note that $\sigma^{XSA}\in\cal T$.

The proof for the private capacity is analogous and follows by replacing $I_\textrm{acc}(X;E)$ with $I(X;E)$.
\end{proof}%\qed

\begin{figure}
\includegraphics[clip=true,width=85mm]{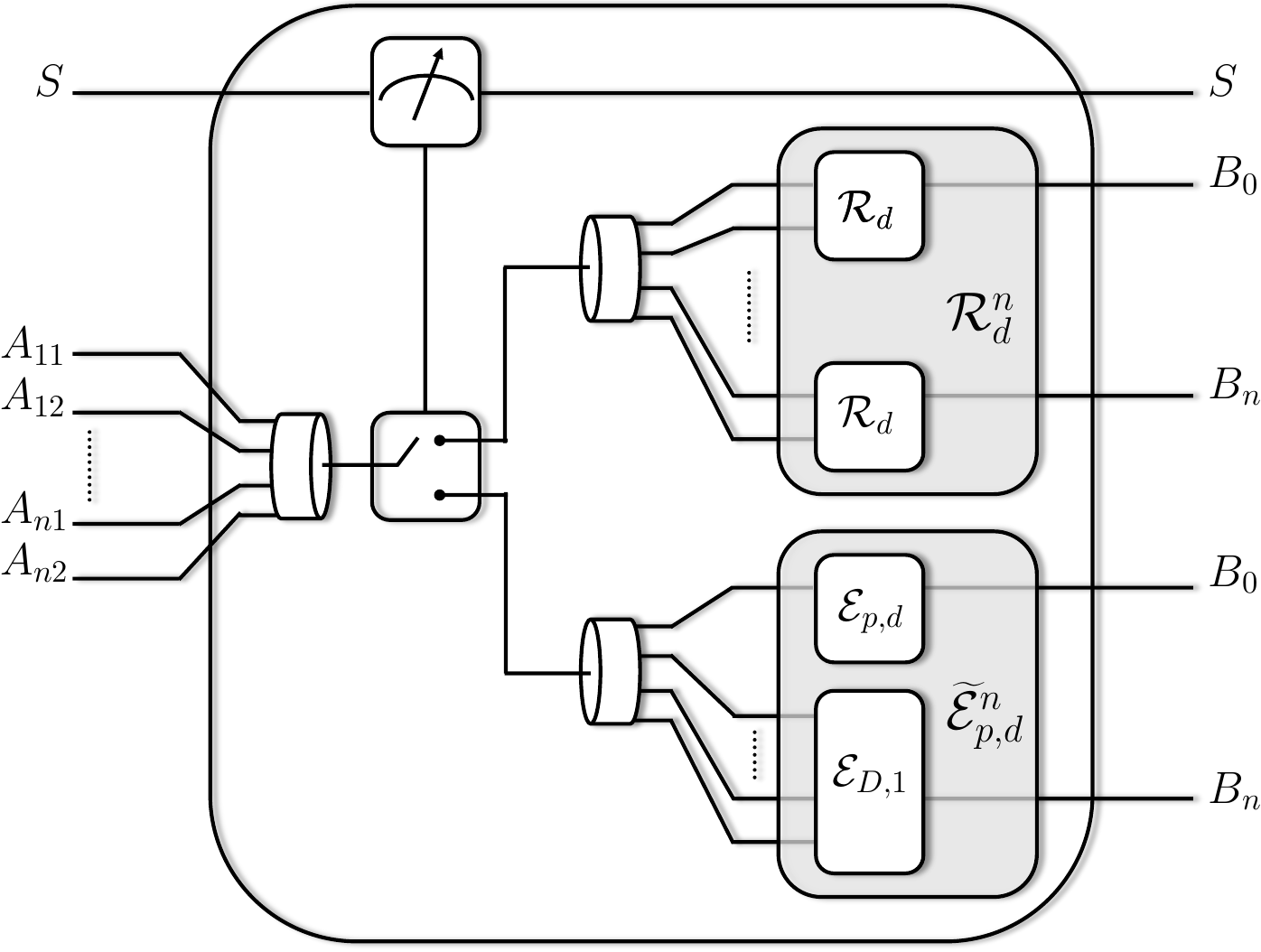}
\caption{The channel has two input registers the control register $S$ and the data register $A$. The control register is measured in the computational basis and depending on the output either the erasure channel $\wt\ce_{p,d}^n$ or $n$ copies of the $d$-dimensional rocket channel are applied.}\label{fig:channel}
\end{figure}

%%%%%%%%%%%%%%%%%%%%%%%%%%%%%%%%%%%%%%%%%%%
There are two types of channels which we will use in place of $\cn_i$. 
The first channel is the erasure channel: 
\begin{equation}
\ce^{A\to B}_{p,d}(\rho_A) = (1-p) \rho_B + p \ketbra{e}{e}_B
\end{equation}
\noindent where $\ketbra{e}{e}$ is the erasure flag and $d$ the dimension of the input register ${A}$. For $p\leq 1/2$ the erasure channel is degradable and ${\cq}(\ce_{p,d})=\cp (\ce_{p,d})=\max\{0, (1-2p)\log d\}$, and $\cc(\ce_{p,d})=(1-p)\log d$.

The locking information of the erasure channel is currently unknown. We might upper bound it by its classical capacity but that gives a too loose bound. A slightly tighter upper bound is obtained in the following lemma (the proof can be found in the Supplemental material):
\begin{lemma}
\label{lem:erasurenlock}
Let $p\leq 1/2$, the locking information of $\ce_{p,d}$ is upper-bounded by
\begin{equation}\label{eq:erasurenlock}
\lw^{(1)}(\ce_{p,d})\le (1-p) \log d - p  \gamma_{d} \log e,
\end{equation}
where $\gamma_d:= \ln d - \sum_{t=2}^d t^{-1}$, and $\lim_{d\to\infty} \gamma_d=\gamma$ is the Euler's constant.
\end{lemma}

For any quantum channel $\cn$ used alongside $\ce_{p,d}$ the classical information is additive:
\begin{lemma}
\label{lem:erasureupper}
For all quantum channels $\cn$ 
\begin{equation}\label{eq:lem3}
\cc^{(1)}\left(\cn\otimes \ce_{p,d}^{\otimes n}\right) = \cc^{(1)}(\cn)+n\cc^{(1)}(\ce_{p,d}).
\end{equation}
\end{lemma}
\begin{proof}
It is trivial that the left hand side of Eq. \ref{eq:lem3} is equal or greater than the right hand side. In order to prove the converse inequality consider the following chain of inequalities:
\begin{align}
\label{eq:erasureuppbound}
\cc^{(1)}(\cn &\otimes \ce_{p,d}^{\otimes n}) = \cc^{(1)}(\mathcal M\otimes \ce_{p,d})\\
&= \max_{\rho} I(X;B_1B_2)\\
&= \max_{\rho} (1-p)I(X;B_1A_2) + pI(X;B_1)\\
&\leq (1-p)\cc^{(1)}(\mathcal M\otimes I) + p\cc^{(1)}(\mathcal M)\\
&= \cc^{(1)}(\mathcal M)+(1-p)\log d\\
&= \cc^{(1)}(\cn)+n(1-p)\log d.
\end{align}
The first equality follows by identifying $\mathcal M$ with $\cn\otimes \ce_{p,d}^{\otimes n-1}$. We let $A_1,A_2$ and $B_1,B_2$ be the input and output of $\mathcal M$ and $\ce_{p,d}$ respectively. The second equality is just the definition of the classical information (see Eq. \ref{eq:ccapacity}). The third equality breaks the mutual information depending on the erasure channel transmitting or erasing. The inequality follows by maximizing each of the two mutual informations individually.  The fourth inequality follows by taking into account that the classical information of the identity is additive and the last one by applying the same argument recursively for $n-1$ times.
\end{proof}

Intuitively, Lemma \ref{lem:erasureupper} states that the erasure channel cannot convey more information than an identity channel of dimension $d^{1-p}$ even in the presence of other channels. Furthermore, we can use the classical capacity to obtain a trivial bound for the locking and private information. It turns out that this trivial bound is tight and is saturated by the channel construction that we introduce below.

%%%%%%%%%%%%%%%%%%%%%%%%%%%%%%%%%%%%%%%%%%%

The second channel that we use alongside $\ce_{p,d}$ is a $d$-dimensional `rocket' channel, $\mathcal R_d$~\cite{Smith_09b}. It consists of two $d$-dimensional input registers $A_1$ and $A_2$ and a $d$-dimensional output register $B$. $A_1$ and $A_2$ are first subject to a random unitary and then jointly decoupled with a controlled dephasing gate. Then, the contents of $A_1$ becomes the output of the channel and the contents of $A_2$ is traced out. Bob also receives the classical description of the unitaries which acted on $A_1$ and $A_2$. Since dephasing occurs after the input registers have been scrambled by a random unitary, it is very hard for Alice to code for such channel, hence it has a very low classical capacity: $\cc(\mathcal R_d)\leq 2$. 

%%%%%%%%%%%%%%%%%%%%%%%%%%%%%%%%%%%%%%%%%%%

Our switch channel construction has the following form:
\begin{equation}
\cn_{n,p,d} = P_0 \otimes \cR_d^n + P_1 \otimes \wt\ce_{p,d}^n\label{mainchannel}
\end{equation}
That is, it allows Alice to choose between $\cR_d^n =\cR_d^{\otimes n}$ and $\wt\ce_{p,d}^n=\ce_{p,d} \otimes \ce_{1,d^{2n-1}}$; a $d$-dimensional erasure channel padded with a full erasure channel to match the input dimension of $\mathcal R_d^n$. 

%%%%%%%%%%%%%%%%%%%%%%%%%%%%%%%%%%%%%%%%%%%
\textbf{Upper bound.} To upper bound the private information of $\cn_{n,p,d}$ we only need to optimize over all the possible different choices of $\cR_d^n$ and $\wt\ce_{p,d}^n$. Thus, the upper bound for $\cp^{(1)}(\cn_{n,p,d}^{\otimes k})$ for $k\ge 1$ reads:
\begin{align}
{\cp^{(1)}}(\cn_{n,p,d}^{\otimes k}&) = \max_{0\le i\le k}  \cp^{(1)}(\ce_{p,d}^{\otimes i}\otimes {(\mathcal {R}_d^n)}^{\otimes k-i})\nonumber\\
&\leq  \max  
\left\{
        \begin{aligned}
        &\left. \cc^{(1)}((\mathcal {R}_d^n)^{\otimes k})\right.\\
        &\max_{1\le i\le k-1}\cc^{(1)}(\ce_{p,d}^{\otimes i}\otimes {(\mathcal {R}_d^n)}^{\otimes k-i}), \nonumber \\
        &\left. \cp^{(1)}(\ce_{p,d}^{\otimes k})\right.
       \end{aligned}
\right.\\
&\leq  \max  
\left\{
        \begin{aligned}
        &\left. 2kn,\right.\\
        &\left(2n+(k-1)(1-p)\log d\right). \\
        &\left(1-2p\right)k\log d 
       \end{aligned}
\right.\label{eq:nupperbound}
\end{align}

%%%%%%%%%%%%%%%%%%%%%%%%%%%%%%%%%%%%%%%%%%%
\textbf{Superadditivity of $\cp^{(1)}$}.  
First, we present the input state such that for $j>i$ uses and for some range of parameters allows to conclude that the private information for $j$ uses is higher than the upper bound~\eqref{eq:nupperbound} for $i$ uses. This state has the form:
\begin{equation}
\label{eq:thestate}
\rho =  \bigotimes_{k=1}^{j-1}\left(\Phi^+_{\wt A_{k}A^{[1]}_{k1}}\otimes\Phi^+_{A^{[1]}_{k2}A^{[k+1]}_{11}}\otimes\sigma_{A}\right)
\end{equation}
where $\Phi^+_{AB} = 1/d\sum_{i,j=1}^d|ii\rangle\langle jj|$. Alice chooses the rocket channel and for the remaining $j-1$ uses of the channel she selects $\ce^n_{p,d}$. We denote with superscript $[k]$ the $k$-th use of the channel and the subscript $ij$ indicates the input register as pictured in Fig. \ref{fig:channel}. The state can be read operationally as follows: Alice keeps the $\wt A_{km}$ registers and sends $A^{[1]}_{k1}$  through the first input of $k$-th $\cR_d$ channel, %(denoted as $\cR_d^{[k]}$), 
$A^{[1]}_{k2}$ through the second %input of $\cR_d^{[k]}$ 
(which will be subsequently discarded by the channel) and $A^{[k]}_{11}$ through $\ce_{p,d}$. The remaining inputs do not play any role, so Alice can send any pure state $\sigma_{A}$ through $\ce_{D,1}$ and $\cR_d^{[k]}$ for $k>j$.  It is easy to verify that:
\begin{equation}\label{privatelowerbound}
\cq^{(1)}( \cn_{n,p,d}^{\otimes j},\rho)=\frac{(j-1)(1-p)}{j}\log d.
\end{equation}

This immediately gives a lower bound for the locking and private information. Now, we are ready to prove Theorem~\ref{mainthm}.

\textbf{Proof}. 
Fix $d=2^{4n^2/(1-2p)}$ and $p=\frac{11}{24}$. Then the regularized upper bounds~\eqref{eq:nupperbound} for $\cp^{(1)}$ after $k$ uses of the channel have the form $U^1_k = \frac{2n}{k}$,  $U_k^2 =\frac{2 n (13 (k-1) n+1)}{k}$ and $U_k^3 = 4 n^2$; the lower bound~\eqref{privatelowerbound} after $k+1$ uses of the channel has the form $L_{k+1}=\frac{26 k n^2}{k+1}$.
 
Consider the differences $D^i_k=-U^i_k+L_{k+1}$ for $i=1,2,3$. Then, a simple substitution shows that $D^1_k=\frac{26 k n^2}{k+1}-\frac{2 n}{k}$, $D^2_k=-\frac{2 n (k-13 n+1)}{k (k+1)}$, $D^3_k=\frac{2 (11 k-2) n^2}{k+1}$. 
All of the differences are positive for $n>k\geq 1$.\qed

The results of the theorem indicate that in order to compute the {\it exact} private capacity of a channel $\cn$ it is necessary to compute $\cp^{(1)}(\cn^{\otimes n})$ for an arbitrary number of uses $n$. In addition, 
we found an example whereby for each $n$ and $1\leq k<n$ having access to one additional copy of the channel up to $n$ provides the parties with the largest possible gain in the capacity, proportional to the output dimension of the channel. Note, that for the channel $\cn_{n,p,d}$ strict superadditivity of both private and coherent information holds for all number of uses of the channel up to $n$. This is markedly different from all previously known channel constructions which exhibit various superadditivity effects. Such constructions exhibited superadditivity for some fixed number of uses of the channel $t$ versus $t+c$ for some $c$. Our construction above shows that the private and coherent information of the \textit{same} channel can be strictly superadditive for an arbitrary number of channel uses.

\textbf{Superadditivity of $\lw^{(1)}$}.
Now we study the conditions necessary to obtain a similar result for the locking information of our channel construction. Some algebra shows that the upper bound given by Lemma \ref{lem:erasurenlock} combined with the lower bound given by~\eqref{eq:thestate} do not yield superadditivity. 
However, the ensembles used for the upper bound of $I(X;A)$ and $I_{acc}(X;E)$ in Lemma \ref{lem:erasurenlock} are different; in particular for the former we used a discrete ensemble of orthogonal states, while the latter bound is obtained via the so called Scrooge ensemble~\cite{Jozsa_94}, in consequence for a given ensemble either Alice should obtain less information or Eve more implying that a bound tighter than that of Lemma \ref{lem:erasurenlock} should hold. In particular, if this tighter bound takes the following form: 

\begin{conjecture}\label{weakconjecture}[Sharper upper bound for $\lw^{(1)}$]
\begin{align}
{\lw^{(1)}}(\ce_{p,d}) &\leq (1-p)\log d -  p\epsilon\log{d}.\label{eq:conj}
\end{align}
where $\epsilon > \frac{1-p}{p(n-1)}$.
\end{conjecture}
The same techniques used in the proof of Theorem \ref{mainthm} would allow to prove superadditivity.

%%%%%%%%%%%%%%%%%%%%
\textbf{Discussion}.
%%%%%%%%%%%%%%%%%%%%
In this paper we have constructed a family of channels for which the private and coherent information can remain strictly superadditive for any  number of uses of the channel. We are able to prove this result by showing that the private information of $k$ uses of the channel is smaller than the coherent information of $k+1$ uses. That is, both quantities can be interleaved use after use for the first $n$ uses of the channel. This shows that even though the quantum capacity is upper bounded by the infinite regularization of the private information, the quantum capacity can be larger than a finite regularization of the private information.
 
The private and locking capacity of a quantum channel characterize its ability to convey classical information securely with two different security criteria. We proved that in order to compute the private capacity it is necessary to consider regularized expressions~\eqref{privatecapformula}. Similarly, we expect weak locking information to be superadditive. For this to be true with our channel construction a tighter bound on the accessible information to the environment would be necessary.

The results shown here raise questions about the properties that a channel has to verify such that its different capacities can be computed exactly using only finitely many (preferably only a few) copies of the channel.

{\bf Acknowledgements:} We thank David Perez Garc\'ia and Maris Ozols for many useful discussions and feedback. 
SS acknowledges the support of Sidney Sussex College. DE acknowledges financial support from the European CHIST-ERA project CQC (funded partially by MINECO grant PRI-PIMCHI-2011-1071) and from Comunidad de Madrid (grant QUITEMAD+-CM, ref. S2013/ICE-2801). This work has been partially supported by the project HyQuNet (Grant No. TEC2012-35673), funded by Ministerio de Econom\'ia y Competitividad (MINECO), Spain. This work was made possible through the support of grant \#48322 from the John Templeton Foundation. The opinions expressed in this publication are those of the authors and do not necessarily reflect the views of the John Templeton Foundation.

\appendix
\section{Proof of Lemma 2}
We denote by $U$ a von Neumann measurement in the basis induced by applying unitary $U$ to the computational basis and by $\left\langle H(U(A)\right\rangle_U$ the entropy averaged over unitaries sampled according to the Haar measure. 
The subentropy of a state $\rho^A$ can be written as:
\begin{equation}
\label{eq:subentropy}
Q(A)=\left\langle H(U(A)\right\rangle_U-H_n \log e
\end{equation}
where $H_n=\sum_{i=1}^n 1/i$ \cite{Jozsa_94}. With Eq. \ref{eq:subentropy} we obtain the following upper bound of the difference between entropy and subentropy:
\begin{align}
\label{eq:subentdif}
H(A)-Q(A) &= H(A) - \left\langle H(U(A)\right\rangle_U + H_n \log e \nonumber \\
              &= \left\langle H(A) - H(U(A)\right\rangle_U + H_n \log e \nonumber \\
              &\leq H_n \log e
\end{align}
the inequality follows since a measurement either leaves entropy unchanged (measurement in the eigenbasis of the state) either increases it.

%\begin{proof} 
For degradable channels the optimization of $\lw$ can be restricted to pure ensembles \cite{Winter_14}, that is, states of the form $\sum_x p_x \ketbra{x}{x}\otimes\ketbra{\phi_x}{\phi_x}$. If we restrict the optimization to states $\rho$ of this form the following chain of inequalities hold:
\begin{align}\label{eq:erasure1lock}
\lw^{(1)}(\ce_{p,d}^{\otimes n}) &= \max_{\rho} I(X;B) - I_\textrm{acc}(X;E)\nonumber\\
                           &= \max_{\rho} (1-p) I(X;A) - p I_\textrm{acc}(X;A)\nonumber\\
                           &\leq \max_{\rho} (1-p) H(A) - p Q(A)\nonumber\\
                           &= \max_{\rho} (1-2p) H(A) + p (H(A)-Q(A)) \nonumber\\
                           &\leq (1-2p) \log d + p H_n \log e \nonumber \\
                           &= (1-p) \log d - p \gamma_{d} \log e. \nonumber
\end{align}
The first inequality follows by bounding the mutual information with the entropy of the input and the accessible information with the subentropy of the input. Subentropy is known to be a lower bound on the accessible information~\cite{Jozsa_94}. The second inequality follows from Eq. \ref{eq:subentdif}.

\end{document}